\documentclass[a4paper,11pt]{amsart}
\usepackage{amsfonts,amssymb,epsfig,latexsym}
\usepackage{amsmath}
\usepackage[latin1]{inputenc}
\usepackage{color,layout}
\usepackage{ifthen}

\usepackage{pdfsync}
\usepackage{graphicx}
\usepackage{color}
\usepackage{pdfsync}
\date{}

\newcommand{\be}{\begin{equation}}
\newcommand{\ee}{\end{equation}}

\def\la{\langle}
\def\ra{\rangle}

\def\R{\mathbb{R}}
\def\C{\mathbb{C}}

\renewcommand{\Re}{\text{{\rm Re}\;}}
\renewcommand{\Im}{\text{{\rm Im}\;}}

\newcommand{\dist}{\text{\rm dist}}

\newtheorem{theorem}{Theorem}[section]
\newtheorem{lemma}[theorem]{Lemma}
\newtheorem{proposition}[theorem]{Proposition}

\theoremstyle{definition}
\newtheorem{definition}[theorem]{Definition}

\newtheorem{remark}[theorem]{Remark}

\numberwithin{equation}{section}

\title{The spectrum of the cubic oscillator}
\author{
Vincenzo GRECCHI${}^1$ \&  
Andr\'e MARTINEZ$ {}^2$
 }

\begin{document}

\maketitle 
\addtocounter{footnote}{1}
\footnotetext{{\tt\small Universit\`a di Bologna,  
Dipartimento di Matematica, Piazza di Porta San Donato, 40126
Bologna, Italy, 
vincenzo.grecchi@unibo.it }\\Partly supported by Universit\`a di Bologna, Funds for Selected Research Topics}  
\addtocounter{footnote}{1}
\footnotetext{{\tt\small Universit\`a di Bologna,  
Dipartimento di Matematica, Piazza di Porta San Donato, 40126
Bologna, Italy, 
andre.martinez@unibo.it }\\Partly supported by Universit\`a di Bologna, Funds for Selected Research Topics}  
\begin{abstract}
We prove the simplicity and analyticity of the eigenvalues of the cubic oscillator Hamiltonian, 
$$
H(\beta)=-\frac{d^2}{dx^2}+x^2+i\sqrt{\beta}x^3,
$$
for $\beta$ in the cut plane $\C_c:=\C\backslash \R_-$. Moreover, we prove that the spectrum consists of the perturbative eigenvalues $\{E_n(\beta)\}_{n\geq 0}$ labeled by the constant number $n$ of  nodes of the corresponding eigenfunctions. In addition, for all $\beta\in\C_c$, $E_n(\beta)$ can be computed as the
Stieltjes-Pad\'e sum of its perturbation series at $\beta=0$. This also gives an alternative proof of the fact that the spectrum of $H(\beta)$ is real when $\beta $ is a positive number. This way, the main results on the repulsive PT-symmetric and on the attractive quartic oscillators are extended to the cubic case.
\end{abstract}  
\baselineskip = 12pt 
\vskip 8cm
{\it Keywords:} Cubic oscillator; Non-selfadjoint operators; Pad\'e approximants; Nodes of eigenfunctions; PT-symmetric operators.
\vfill\eject
\section{Introduction}
The real cubic oscillator,  \be H=-\frac{d^2}{dx^2}+x^2+x^3,\label{HQR}\ee
has been considered from the very beginnings of quantum mechanics 
as one of the simplest operator to study (see, e.g.,  \cite{DA}). Indeed, it is the simplest model after the two cases of a linear potential and a quadratic potential. But, in contrast with these two models (for which Airy and Weber functions can be used), no general analogous special functions are known for the cubic case (see,  however, \cite{CGM}).
In addition, from a quantum point of view, it appears that there exists an infinity
  of self-adjoint extensions from $C_0^\infty (\R)$, all  with discrete spectrum (see, e.g., \cite{RS} for such considerations). Therefore, this apparent simplicity actually hides important difficulties, for which a rigorous treatment is needed.\\

Here, in order to get a better understanding of the problem, we consider the complex Hamiltonian,
\be H(\beta)=-\frac{d^2}{dx^2}+x^2+i\sqrt{\beta}x^3,\label{THITCP}\ee
where $\beta\not= 0$ is a complex parameter with $|\arg \beta|<\pi$.
In particular, for positive $\beta$, $H(\beta)$ is  PT-symmetric, as the repulsive quartic oscillators defined and studied in \cite{BG}. \\

Strangely enough, until recently  there were rather few rigorous published   results concerning this operator. Among them, however,
we can quote the  analyticity of the eigenvalues for small $|\beta|$ and the Borel summability \cite{CGM}, extended to the boundary $\arg(\beta)=\pm\pi,$ as distributional Borel summability \cite{C}.\\

On the contrary, many papers have been devoted to the quartic anharmonic oscillator (see, e.g., \cite{BG, LM, LM1, Sim, Vo}), for which a rather complete series of results have been proved.\\

Indeed, a renewed interest for this model came only after the Bessis-Zinn Justin conjecture in 1992 (successively extended to other oscillators by Bender-Boettcher in \cite{BB}, after the results of \cite{BG} on the quartic oscillator), claiming that  the spectrum of $H(\beta)$ should be real for positive $\beta$.
After an important step due to Dorey-Dunning-Tateo in \cite{DDT}, the conjecture was finally proved by Shin in \cite{Sh} (actually for a more general class of Hamiltonians).\\  

In this paper,  we recover Shin's result  in our case, but  with a completely different proof. In addition, our proof permits us to give a complete information on the spectrum of $H(\beta)$ for all $\beta$ in the cut plane.\\

More precisely, we prove that the spectrum of $H(\beta)$ consists of the perturbative simple eigenvalues $\{E_n(\beta)\}_{n\geq 0}$, labeled by the constant number $n$ of  the nodes of the corresponding eigenfunctions, that is,  its zeros lying in the lower half-plane (that are stable in the $\beta =0$ limit).  
In addition, we prove that, for all $\beta$ such that $|\arg\beta|<\pi$, each eigenvalue $E_n(\beta)$ is the
Stjelties-Pad\'e sum of its perturbation series at $\beta=0$. In this way we extend to the cubic case the results by Loeffel-Martin-Simon-Wightman \cite{LM} on the quartic oscillator (see also \cite{Sim}).\\

Our method, that we believe to be new,  is based on a combined use of perturbative and complex semiclassical arguments. In particular, it exploits in an intensive way the stability of the nodes  of the eigenfunctions when the energy tends to infinity. More precisely we prove that, for any eigenfunction depending continuously on $\beta$, its number of nodes is indeed independent of $\beta$ and remains uniformly bounded. On the other hand,  semiclassical considerations show that this number is related to the value of the energy and, in particular, tends to infinity as the energy tends to infinity. This shows that the energy cannot explode for finite $\beta$ (even in the limit $\arg\beta\to \pm\pi$), and permits us to obtain global results. A similar strategy has been used by Loeffel and Martin \cite{LM1} for the anharmonic quartic oscillator, with the difference that, in their case, the particular form of the potential allowed the use of variational inequalities, leading to the result in a much simpler way.\\

Let us observe that the limit of $E_n(-b-i\epsilon)$ as $\epsilon\to 0^+,$ gives a generalized resonance of (\ref{HQR}) defined as the limit of resonances of analytically regularized Hamiltonians \cite{CM}.\\

A general discussion on the problem, extended to the non-modal solutions, has also been presented in two recent papers in collaboration with Marco Maioli \cite{GMM1,GMM2} (but the control on the non-perturbative levels was not complete).\\

Finally, in relation with our result, we would like to mention the complex semiclassical studies by Eremenko-Gabrielov-Shapiro \cite {EGS}, Delabaere-Pham \cite{DP1, DP2} and  Delabaere-Trinh \cite{DT} (see also \cite{Tr} and references therein), and the numerical computations done
by Bender-Weniger \cite{Be} and G. Alvarez \cite{Al}.\\

In the next section, we describe our main results and make some general remarks on them. Section 3 is devoted to some preliminaries results concerning any arbitrary eigenvalue $E(\beta)$ of $H(\beta)$, depending continuously on $\beta$ in some open set $\Omega\subset \C_c$. In particular, it is shown that the corresponding eigenfunction admits a finite number of zeros (nodes) in the half-plane $\{ \Im x <0\}$, and that this number is independent of $\beta$. Using this property, we show that $E(\beta)$ can be continued (as an eigenvalue of $H(\beta)$) along any simple path of $\C_c$. In Section 4, we consider the particular case where the path relates $\Omega$ with 0, and we prove that, when $\beta$ becomes close to 0, then $E(\beta)$ can be identified with some perturbative eigenvalue $E_n(\beta)$. The simplicity of the eigenvalues is proved in Section 5 where, still using the properties of the nodes of the eigenfunctions, we show that each $E_n(\beta)$ can be continued in a holomorphic way to the whole cut plane. Then, in order to have the Stieltjes property, we study the behavior of $E_n(\beta)$ both when $\beta$ becomes close to the cut and when $|\beta|$ tends to infinity . In particular, we show that $E_n(\beta)$ admits a finite limit $E_n^\pm (-b)$ as $\beta\to -b\in (-\infty, 0)$ with $\pm\arg\beta <\pi$ (Section 6), and that $E_n(\beta)$ behaves like $\beta^{1/5}$ as $|\beta|\to \infty$ (Section 7). This permits us, in Section 8, to prove that the function $\beta^{-1}(E_n(\beta)-E_n(0))$ is Stieltjes, and to complete the proof of our results.

\section{Results}

We study the spectrum of the non selfadjoint operator,
\be
H(\beta):= -\frac{d^2}{dx^2} + x^2 +i\sqrt{\beta}x^3
\ee
on $L^2(\R)$, for $\beta \in \C_c:=\C\backslash (-\infty , 0]$ and with domain,
$$
{\mathcal D}:=\{ u\in H^2(\R)\, ;\, x^3 u(x) \in L^2(\R)\},
$$
where $H^2(\R)$ stands for the usual Sobolev space of index 2 on $\R$.\\

It is well known (see, e.g., \cite{CGM}) that, for $\beta\in\C_c$, $H(\beta)$ forms an analytic family of type A (in the sense of Kato \cite{K}) with compact resolvent, and that, for any $\theta\in (-\pi, \pi)$, the operator $H(be^{i\theta})$ tends to $H_0:= -\frac{d}{dx^2} + x^2$ (with domain ${\mathcal D}_0:=\{ u\in H^2(\R)\, ;\, x^2 u(x) \in L^2(\R)\}$) in the norm-resolvent sense, as $b\to 0_+$.\\

Moreover, denoting by $E_n(0) = 2n+1$ ($n\geq 0$) the $(n+1)$-th eigenvalue of $H_0$, one can see (\cite{CGM}, Theorem 2.13 and \cite{C}) that, for all $n\geq 0$, there exists $\delta_n>0$ such that, for any $\beta\in \C_c$ with $|\beta| < \delta_n$, one has,
\be
\label{calic}
\sigma (H(\beta))\cap\{ |z-E_n(0)|<\delta_n\} = \{E_n(\beta)\},
\ee
where $E_n(\beta)$ is a simple eigenvalue depending analytically on $\beta$ and bounded on $ \C_c\cap \{|\beta| < \delta_n\}$. Finally, $E_n(\beta)$ admits an asymptotic expansion as $\beta\to 0$ ($\beta\in\C_c$), of the form,
\be
\label{calicbis}
E_n(\beta)\sim E_n(0) + \sum_{k\geq 1}e_{n,k} \beta^k,
\ee
where all the coefficients $e_{n,k}$ are real, are given by the Rayleigh-Schr\"odinger perturbation theory, and satisfy to the estimate,
\be
\label{anasymb}
|e_{n,k}|\leq D_n C_n^{k}\, k!\quad (k\geq 0),
\ee
where $C_n$ and $D_n$ are positive constants. (Actually, it is proved in \cite{CGM, C} that the previous series is Borel summable to $E_n(\beta)$ when $|\arg\beta |<\frac{3\pi}4-\varepsilon$, $\varepsilon >0$ arbitrary, $|\beta|$ small enough, and more generally, in the distributional sense for the other values of $\arg\beta$; however, here we will not use these results.)\\

Here, we prove,
\begin{theorem}\sl
\label{mainth}
For all $\beta\in\C_c$, the spectrum $\sigma (H(\beta))$of $H(\beta)$ consists  of simple eigenvalues $E_n (\beta)$ ($n=0,1,...$) depending analytically on $\beta \in \C_c$. Moreover, for all $\beta\in\C_c$ and $n\geq 0$, one has,
\be
E_n(\beta) = E_n (0) +\beta \lim_{j\to +\infty}\frac{P_{n,j}(\beta)}{Q_{n,j}(\beta)},
\ee
where $P_{n,j}$ and $Q_{n,j}$ are the polynomials of degree $j$ (``diagonal Pad\'e approximants'') defined by,  
$$
\left\{
\begin{array}{l}
Q_{n,j}(0)=1;\\
\left| \frac{P_{n,j}(\beta)}{Q_{n,j}(\beta)} -\sum_{k=1}^{2j}e_{n,k}\beta^k\right| ={\mathcal O}(|\beta|^{2j+1})\quad\mbox{as } \,\, \beta\to 0.
\end{array}\right.
$$
\end{theorem}
\vskip 0.5cm

Actually, one point of the proof will consist in showing that the series $(|e_{n,k}|)_{k\geq 1}$ is Stieltjes (see, e.g., \cite{Wa} for a definition). In particular, its diagonal Pad\'e approximants  satisfying the condition $Q_{n,j}(0) =1$ are well defined and unique (see \cite{Wa}, Chap. XX). Indeed, we also prove,\\
\begin{theorem}\sl
\label{mainth2}
For all $\beta\in\C_c$, one has,
$$
E_n(\beta)=E_n(0) + \beta\int_0^{+\infty}\frac{\rho_n(t)}{1+\beta t} dt,
$$
where $\rho_n$ is a real-analytic positive function on $(0,+\infty)$,  such that  $t^{1/5}\rho_n(t)$ admits a limit as $t\to 0_+$, and,
$$
\ln \rho_n(t) = -\frac{8}{15}\, (t+{\mathcal O}(\ln t)),\quad \mbox{as } t\to +\infty.
$$
Moreover, the measure $\rho_n(t)dt$ is the only solution of the moment problem,
$$
\int_0^{+\infty} t^k\rho_n(t)dt =|e_{n, k+1}|\quad (k=0,1,\dots).
$$
In particular, the constant $C_n$ in (\ref{anasymb}) can be taken arbitrarily close to $15/8$.

\end{theorem}
\vskip 0.3cm

\begin{remark}\sl By the results of Sibuya \cite{Si}, we already know that, for $\beta\in\C_c$ fixed, the large enough eigenvalues of $H(\beta)$ are simple. However, no information is given in \cite{Si} about how much large they must be (in particular when $\beta$ becomes close to infinity or to the boundary of $\C_c$). Moreover, still by \cite{Si}, one can see that any eigenspace Ker$(H(\beta) -E)$ is one dimensional (that is, the geometric multiplicity of $E$ is one). Therefore, the possible (algebraic) multiplicity of an eigenvalue necessarily means the appearance of non vanishing Jordan blocks.
\end{remark}
\begin{remark}\sl
When $\beta$ is a positive number, the operator $H(\beta)$ is PT-symmetric and the reality of its spectrum 
is proved by Shin in \cite{Sh}. Here, we recover this result by observing that, since all the $e_{n,k}$ are real, so are all the coefficients of $P_{n,j}$ and $Q_{n,j}$. Actually, the reality of $E_n(\beta)$ is indeed a consequence of its simplicity which, in turns, comes out from a good control on the nodes of the associated eigenfunction (see Definition \ref{defnode} and Proposition \ref{nbrenode}).
\end{remark}

\vskip 0.3cm

\begin{remark}\sl
In particular, it follows from this theorem that all the eigenvalues of $H(\beta)$ are  holomorphic extensions to the whole cut plane $\C_c$ of the perturbative eigenvalues computed near $\beta =0$.
\end{remark}

\section{Preliminaries}\label{prelim}
We start with some general results concerning the eigenvalues of $H(\beta)$.\\

By the results of \cite{K} and \cite{Si}, we know that, if $E_0$ is an eigenvalue of $H(\beta_0)$ for some $\beta_0\in\C_c$, then its multiplicity $m$ is finite and there exits an integer $m'\in [1, m]$ and a $m'$-valued
analytic function $\beta\mapsto E(\beta)$ defined on a neighborhood of $\beta_0$ with branch point at $\beta_0$,  such that $E(\beta)\rightarrow E_0$ as $\beta\to\beta_0$, and $E(\beta)$ is an eigenvalue of $H(\beta)$ for all $\beta$ in this neighborhood.\\

In other words, the function $\beta \mapsto E(\beta)$ is well-defined and analytic on the $m'$-sheets covering space over $\Omega\backslash\{\beta_0\}$, where $\Omega$ is a neighborhood of $\beta_0$ in $\C_c$, and moreover it is continuous at $\beta_0$. \\

 Let us note that, by the results of \cite{Si}, the algebraic multiplicity of $E_0$ coincides with its multiplicity as a zero of the Stokes-Sibuya coefficient (see, e. g., \cite{Tr}). In particular, the finitude of $m$ is a direct consequence of the global analyticity of this coefficient with respect to the energy.\\

The purpose of  this section is to continue holomorphicaly any of the branches of $E(\beta)$, in order to reach any point of $\C_c$ (in particular, points arbitrarily close to 0). At first, we prove,
\begin{proposition}\sl
\label{contfonctpropre}
There exists a $L^2(\R)$-multiple-valued continuous function $\psi_\beta$ of $\beta\in \Omega$,  such that, for all $\beta\in \Omega$, one has  $H(\beta)\psi_\beta = E(\beta)\psi_\beta$ and $\Vert \psi_\beta\Vert =1$.
\end{proposition}
\begin{proof} For $E\in\C$, denote by $\psi_{{\pm},E}$ the two solutions of $H(\beta)\psi_{{\pm},E}= E\psi_{{\pm},E}$ that are subdominant near $\R_\pm$ respectively, and are given by the Sibuya asymptotics (see \cite{Si}, Chapter 2). In particular,  $\psi_{{\pm},E}$ are $L^2(\R_\pm)$-valued analytic functions of $(E,\beta)$, and the fact that $E(\beta)$ is an eigenvalue of $H(\beta)$ means that there exists a complex number $\alpha_\beta$ such that $\psi_{+,E(\beta)} = \alpha_\beta \psi_{-,E(\beta)}$. Moreover, since $\psi_{-,E}(z)$ is a non identically zero entire function  of $z$, for any nonempty open set $\omega\subset\C$, one has $\Vert \psi_{-,E}\Vert_{L^2(\omega)} \not=0$.
Thus, one can write $\alpha_\beta =\la\psi_{+,E(\beta)},\psi_{-,E}\ra_{L^2(\omega)}/\Vert \psi_{-,E}\Vert_{L^2(\omega)}^2$, showing that $\alpha_\beta$ is a continuous  function of $\beta$. As a consequence, so is $\Vert \psi_{+,E(\beta)}\Vert^2 = \Vert \psi_{+,E(\beta)}\Vert^2_{\R_+}+|\alpha_\beta|^2\Vert \psi_{-,E(\beta)}\Vert^2_{\R_-}$, and the function $\psi_\beta:= \psi_{+,E(\beta)}/\Vert \psi_{+,E(\beta)}\Vert$ solves the problem.
\end{proof}
\begin{remark}\sl Actually, by the results of \cite{Si}, we know that $\alpha_\beta$ can be expressed as $\alpha_\beta= C(\sqrt{\beta}, E(\beta))$, where the function $\C^2\ni (\alpha, E)\mapsto C(\alpha, E)$ is entire. 
\end{remark}

Next, we prove,
\begin{proposition}\sl
\label{zero1}
For all $\beta\in\Omega$, the eigenfunction $\psi_\beta$ does not admit any zero in the strip,
$$
S_\beta:= \left\{ x\in\C\, ;\, 0\leq \Im x \leq \frac 2{3\sqrt{|\beta|}}\cos(\frac{{\rm arg}\beta}2)\right\}.
$$
\end{proposition}
\begin{proof} Setting $V_\beta(x):= x^2 +i\sqrt{\beta}x^3$, and using the equation, we see that, for any $s,t\in\R$, we have,
\begin{eqnarray}
\label{psi'psi}
\Im \psi'_\beta (s+it)\overline{\psi_\beta (s+it)} &=& \int_{-\infty}^s \Im (V_\beta (r+it)-E(\beta))|\psi_\beta (r+it)|^2 dr\nonumber\\
&=&- \int_s^{+\infty} \Im (V_\beta (r+it)-E(\beta))|\psi_\beta (r+it)|^2 dr.
\end{eqnarray}
Therefore, it is enough to prove that, for $0\leq t\leq \frac 2{3\sqrt{|\beta|}}\cos(\frac{{\rm arg}\beta}2)$, the function $r \mapsto \Im (V_\beta (r+it)-E(\beta))$ changes sign at most once on $\R$ (in that case, at least one -- and thus both -- of the two previous integrals is necessarily different from zero). Writing $\beta =be^{i\theta}$ with $b>0$ and $|\theta|<\pi$, we compute,
\begin{eqnarray*}
\frac{d}{dr}\Im V_\beta (r+it)&=&\frac{d}{dr}\left(2rt +r^3\sqrt{b}\cos\frac{\theta}2 -3r^2t\sqrt{b}\sin\frac{\theta}2-3rt^2\sqrt{b}\cos\frac{\theta}2 \right)\\
&=& 3r^2\sqrt{b}\cos\frac{\theta}2-6rt\sqrt{b}\sin\frac{\theta}2+2t - 3t^2\sqrt{b}\cos\frac{\theta}2.
\end{eqnarray*}
This is a second order polynomial function of $r$, with (reduced) discriminant given by,
$$
\Delta' = 9t^2b-6t\sqrt b\cos\frac{\theta}2=3t\sqrt{b}(3\sqrt{b} t - 2\cos\frac{\theta}2).
$$
In particular, this discriminant is non positive when $0\leq t\leq \frac2{3\sqrt b}\cos\frac{\theta}2$, and the result follows.
\end{proof}
\begin{proposition}\sl
\label{zero2}
There exists $C(\beta)>0$, depending continuously on $\beta\in\Omega$, such that the eigenfunction $\psi_\beta$ does not admit any zero in the set,
$$
\{ x\in \C \, ;\, \Im x <0\, ,\, |x| \geq C(\beta)\}.
$$
\end{proposition}
\begin{proof} This is just an immediate consequence of the Sibuya asymptotics at infinity of $\psi_\beta (x)$, in the two sectors,
\be
{\mathcal S}^\pm:= \{ x\in\C\, ;\, \left|\arg(ix) +\frac{\arg\beta}{10} \mp\frac{2\pi}5\right| < \frac{3\pi}{5}\}, \,\,
\ee
(see \cite{Si}, Chapter 2). Since ${\mathcal S}^{+}\cup {\mathcal S}^{-} = \C\backslash (e^{i(\frac\pi{2}+\frac\theta{10})}\R_+)\supset \{ \Im x < 0\}$, the result follows. 
\end{proof}
\vskip 0.3cm

We deduce from Propositions \ref{zero1} and \ref{zero2} that, for any continuous function $\tilde C(\beta)\geq C(\beta)$, the number of zeros of $\psi_\beta$ lying in the region,
$$
\tilde\Gamma_\beta:= \{ \Im x \leq \frac 2{3\sqrt{|\beta|}}\cos(\frac{{\rm arg}\beta}2)\, ,\, |x|\leq \tilde C(\beta)\}
$$
is constant. Indeed, the two previous propositions show that the boundary $\partial\tilde\Gamma_\beta$ of $\tilde\Gamma_\beta$ does not contain zeros, and thus the  number of zeros in $\tilde\Gamma_\beta$ (finite, because of the analyticity of $\psi_\beta$) is given by,
\be
\label{nbrezeros}
N_\beta = \frac1{2i\pi}\oint_{\partial\tilde\Gamma_\beta} \frac{\psi_\beta'(x)}{\psi_\beta(x)}dx=\frac1{2i\pi}\oint_{\partial\Gamma_\beta} \frac{\psi_\beta'(x)}{\psi_\beta(x)}dx <\infty.
\ee
In particular, this number is a continuous function of $\beta$, and therefore is constant for $\beta\in\Omega$. Moreover, still by Propositions \ref{zero1} and \ref{zero2}, these zeros are the only ones in $\{ \Im x \leq \frac 2{3\sqrt{|\beta|}}\cos(\frac{{\rm arg}\beta}2)\}$. At this point,  it will be useful to set,
\begin{definition}\sl
\label{defnode}
For any  eigenfunction $\psi$ of $H(\beta)$, we call {\bf node} of $\psi$ any zero of $\psi$ that lie in the lower half-plane $\C_-:=\{ \Im x<0\}$, and we denote by $N(\psi)$ their total number.
\end{definition}
The previous discussion shows that we have,
\begin{proposition}\sl
\label{nbrenode}
$N(\psi_\beta)$ is finite and does not depend on  $\beta\in\Omega$. Moreover, the nodes of $\psi_\beta$ coincide with its zeros  in the domain $\{ \Im x \leq \frac 2{3\sqrt{|\beta|}}\cos(\frac{{\rm arg}\beta}2)\}$, and their total number can be computed by using the formula  (\ref{nbrezeros}). 
\end{proposition}
\vskip 0.3cm
From now on, we denote by $N_{\psi}$ the total number of nodes  of $\psi_\beta$ for any  $\beta\in\Omega$.\\

In the following, we will also need a better control of the zeros in the region $\{ \Im x >0\}$. We have,\
\begin{proposition}\sl
\label{nozerosecteur}
The function $\psi_\beta$ does not admit any zero in the region,
$$
\{ 0< \arg x<\frac{\pi -\arg \beta}{10}\} \,\cup\, \{ \pi -\frac{\pi +\arg \beta}{10}< \arg x<\pi\}.
$$
\end{proposition}
\begin{proof} By (\ref{psi'psi}) for $s=t=0$, we see that $\Im \psi'_\beta (0)\overline{\psi_\beta (0)}<0$. Moreover, since $\psi_\beta$ is subdominant in both sectors $S_{\pm 1}(\beta)$ defined by,
\be
\label{secteurs}
S_j(\beta):= \{ x\in\C\, ;\, \left|\arg(ix) +\frac{\arg\beta}{10} -\frac{2j\pi}5\right| < \frac{\pi}{5}\}, \,\, j= \pm 1
\ee
we see that, for $\alpha\in (0, (\pi -\arg \beta)/10)$, the function $s\mapsto \psi_\beta (se^{i\alpha})$ is in $L^2(\R_+)$. Therefore, for $s>0$, using the equation  $H(\beta)\psi_\beta = E(\beta)\psi_\beta$, we obtain,
\begin{align*}
\Im \psi'_\beta (se^{i\alpha})\overline{\psi_\beta (se^{i\alpha})}&=\Im \psi'_\beta (0)\overline{\psi_\beta (0)}+\int_0^s p_\alpha (r)|\psi_\beta (re^{i\alpha})|^2 dr \\
&=-\int_s^{+\infty}p_\alpha (r)|\psi_\beta (re^{i\alpha})|^2 dr,
\end{align*}
with, 
\begin{align*}
p_\alpha (r):&= \Im (r^2e^{4i\alpha} +i\sqrt{\beta} r^3e^{5i\alpha}-E(\beta)e^{2i\alpha})\\
&=r^2\left[\sin 4\alpha + r\cos(5\alpha +\frac{\arg\beta}2)\right]-\Im (E(\beta)e^{2i\alpha}).
\end{align*}
In particular, since both $\sin 4\alpha$ and $\cos(5\alpha +\frac{\arg\beta}2)$ are positive, $p_\alpha(r)$ changes sign at most once on $\R_+$, and is positive for large $r$. We deduce that $\Im \psi'_\beta (se^{i\alpha})\overline{\psi_\beta (se^{i\alpha})}<0$ for all $s>0$. \\
In the same way, one finds $\Im \psi'_\beta (se^{i\alpha})\overline{\psi_\beta (se^{i\alpha})}<0$ when $s<0$ and $\pi-\alpha \in (0, (\pi +\arg\beta )/10)$, and the result follows.
\end{proof}

Now, we plan to extend continuously the function $E(\beta)$ as much as possible, as an eigenvalue of $H(\beta)$. We first show,
\begin{proposition}\sl
\label{bounded}
The function $\Omega\ni \beta \mapsto E(\beta)$ is uniformly bounded.
\end{proposition}
\begin{proof} By absurd, assume there exists a sequence $(\beta_k)_{k\geq 0}$ in $\Omega$ such that $|E(\beta_k)|\to\infty$ as $k\to\infty$. By extracting a subsequence, we can also assume that $(\beta_k)_{k\geq 0}$ admits a limit $\beta_\infty=b_\infty e^{i\theta_\infty}\in\partial\Omega$, with $b_\infty >0$ and $\theta_\infty \in [-\pi, \pi]$.

In the equation $H(\beta_k)\psi_{\beta_k}(x) = E(\beta_k)\psi_{\beta_k}(x)$, we make the (complex) change of variable,
$$
x= |E(\beta_k)|^{1/3}e^{-i\theta_k/10}y,
$$
where we have used the notation $\theta_k:=\arg\beta_k \in (-\pi, \pi)$.
Then, setting,
\be
\label{hlambda}
h:=  |E(\beta_k)|^{-5/6}\quad ; \quad \lambda:= \frac{E(\beta_k)e^{-i\theta_k/5}}{|E(\beta_k)|}\, ;
\ee
$$
P:=-h^2\frac{d^2}{dy^2}+ i\sqrt{|\beta_k|}y^3+h^{2/5}e^{-i\theta_k/5}y^2,
$$
we see that the new function $\varphi (y):= \psi_{\beta_k}(|E(\beta_k)|^{1/3}e^{-i\theta_k/10}y)$ is solution of,\\
\be
P\varphi (y) = \lambda \varphi (y).
\ee
{}\\
Moreover, since $\psi_\beta$ is subdominant in both sectors $S_{\pm 1}(\beta)$ defined in (\ref{secteurs}),
 we observe that the function $\varphi$ is automatically  exponentially small at infinity on $\R$ (and, actually, even in the limit $|\arg\theta_k |\to \pi$). In particular, it is an eigenfunction of $P$ in the usual sense, and, as $k\to\infty$, we have $h\to 0_+$, $|\beta_k|\to b_\infty >0$, $\theta_k\to\theta_\infty$, and $\lambda$ stays on the circle $\{|\lambda | =1\}$, with $\Re\lambda >0$ (since $\Re P= -h^2\frac{d^2}{dy^2} + h^{2/5}y^2\cos (\theta_k/5)\geq h^{6/5}\sqrt{\cos(\theta_k/5)} >0$). By extracting a subsequence, we can also assume that $\lambda =\lambda (h)$ admits a limit $\lambda_0$.\\
 
 Then, we can try to use the standard  semiclassical complex WKB method, in order to obtain another estimate on the number of nodes of  $\psi_{\beta_k}$. \\

 Since $\Re\lambda >0$, we see that, among the three turning points (that is, the complex solutions of $i\sqrt{|\beta_k|}y^3+h^{2/5}e^{-i\theta_k/5}y^2=\lambda$), exactly two of them (say, $y_\pm(\lambda)$) are in the half-plane $\{ \Im y < Ch^{2/5}\}$ ($C>0$ constant) and are respectively close to $b_\infty^{-1/6}e^{i(2\arg\lambda -\pi)/6}$ and $b_\infty^{-1/6}e^{i(2\arg\lambda -5\pi)/6}$, while the third one $\tilde y(\lambda)$ is close to $b_\infty^{-1/6}e^{i(2\arg\lambda +3\pi)/6}$.\\
 
  Moreover, each of this points gives rise to three anti-Stokes lines, delimiting three complex open sectors. In particular, the sectors attached either to $y_+(\lambda)$ or to $y_-(\lambda)$ are six. In addition, any anti-Stokes line is either bounded (in which case it relates two different turning points) or asymptotic at infinity to one of the five Sibuya directions, given by $\{ \arg y = (4j-3)\pi/10\}$, $j=1,\dots,5$.
  
 Then, for topological reasons related to the properties of the anti-Stokes lines (in particular, the fact that they cannot cross each other away from the turning points), we see that at least one of these six sectors necessarily contains either $\{ y\in \R_{+}\, ;\, y >>1\}$ or $\{ y\in \R_{-}\, ;\, y<< -1\}$, without containing any other turning point.\\
  
  Let us assume, for instance, that this open sector (that we denote by $\Sigma_+$) is attached to $y_+(\lambda)$ and contains $\{ y\in \R_{+}\, ;\, y >>1\}$ (the other cases are similar). Then, the complex WKB method tells us that $\varphi (y)\not =0$ on $\Sigma_{+}$, and that one has the semiclassical asymptotics,\\
\be
\label{BKW}
\frac{\varphi'(y)}{\varphi (y)} = \frac{-1}{h}\sqrt{i\sqrt{|\beta_k|}y^3+h^{2/5}e^{-i\theta_k/5}y^2-\lambda} +{\mathcal O}(1) \quad \mbox{as}\,\, h\to 0_+,
\ee
{}\\
 locally uniformly with respect to $y \in \Sigma_+$. Moreover, this asymptotics remains valid in the two  sectors adjacent to $\Sigma_+$, as long as one do not cross a third anti-Stokes line containing a turning point. In particular, it ceases to be valid at the boundary, say $L_{+}$, between these two adjacent sectors (near which the zeros of $\varphi$ lie). However, there exists a general estimate on $\varphi' /\varphi$, valid near $L_{+}$, too.  We describe it in the next lemma.\\
 
For any $\delta >0$ fixed,  set,
 $$
 {\mathcal E}_\delta := \{ y\in\C\, ; \, \dist (y, \varphi^{-1 }(0)) \geq \delta h\}.
 $$
 Then, we have,\\
 \begin{lemma}\sl
 \label{estphi'surphi}
 For any $\delta >0$, one has,
  $$
 \frac{\varphi'(y)}{\varphi (y)} = {\mathcal O}\left(\frac1{h}\right),
 $$
 locally uniformly for $ y\in {\mathcal E}_\delta$, and uniformly for $h>0$ small enough.
 \end{lemma}
This result is probably well-known to semiclassical experts, but for the sake of completeness we give a proof in the Appendix. 
 \vskip 0.3cm
 Now, we  estimate the number of nodes of $\varphi$ lying near $y_+(\lambda)$. By Propositions \ref{nbrenode} and \ref{nozerosecteur}, we already know that this number is less than $N_\psi$ and is bounded uniformly with respect to $h$. As a consequence, denoting by $W_+$  a sufficiently small neighborhood of $y_+(\lambda)$, for any constant $\delta >0$ arbitrarily small, and for all $h$ small enough, there exists a point $y_h\in L_+\cap {\mathcal W}_+$ with $(1+N_\psi )\delta\leq |y_h -y_+(\lambda)|\leq (1+2N_\psi )\delta$ and $ \dist (y, \varphi^{-1 }(0)) \geq \delta$. We denote by $D$ the disk centered at $y_+(\lambda)$ with radius $|y_h -y_+(\lambda)|$ (in particular, if $W_+$ has been chosen small enough, the part of $L_+$ between $y_+(\lambda)$ and $y_h$ is included in $D$, and the boundary  $\partial D$  of $D$ is transverse to $L_+$ at $y_h$). \\
 
Then, by construction, we have $\varphi (y) \not= 0$ when $y\in\partial D$, and  the number of nodes of $\varphi$ lying in $D$ can be written as,
$$
N_D = \frac1{2i\pi}\oint_{\partial D}\frac{\varphi'(y)}{\varphi (y)} \, dy.
$$
 For any $\varepsilon >0$ arbitrarily small, we set,
 $$
 \gamma (\varepsilon ):= \partial D \cap \{ |y-y_h| \leq \varepsilon\}\, ;\,  \Gamma (\varepsilon )=\partial D \backslash  \gamma (\varepsilon ).
 $$
 Since the length of $\gamma (\varepsilon)$ is ${\mathcal O}(\varepsilon)$, by using Lemma \ref{estphi'surphi}, we obtain,
\be
\label{ND1}
 N_D = \frac1{2i\pi}\oint_{\Gamma (\varepsilon )}\frac{\varphi'(y)}{\varphi (y)} \, dy +{\mathcal O}(\frac{\varepsilon}{h}),
\ee
 uniformly with respect to $\varepsilon$ and $h$. On the other hand, since $\dist (\Gamma (\varepsilon ), L_+)\geq \varepsilon$, on this set we can use the asymptotics (\ref{BKW}), leading to,
\be
\label{ND2}
 \oint_{\Gamma (\varepsilon )}\frac{\varphi'(y)}{\varphi (y)} \, dy = \oint_{\Gamma (\varepsilon )}\frac{-1}{h}\sqrt{i\sqrt{|\beta_k|}y^3+h^{2/5}e^{-i\theta_k/5}y^2-\lambda}\, dy +{\mathcal O}_\varepsilon(1),
\ee
where the notation ${\mathcal O}_\varepsilon(1)$ means that this quantity depends on $\varepsilon$, and is uniformly bounded as $h\to 0_+$.\\

Now, by extracting a subsequence, we can also assume that  $y_h$ admits a limit $y_0$, as $h\to 0_+$. Hence, gathering this information with (\ref{ND1})-(\ref{ND2}), we obtain,
$$
N_D = \frac1{2i\pi}\oint_{\Gamma_0}\frac{-1}{h}\sqrt{i\sqrt{b_\infty}y^3-\lambda_0}\, dy+o_\varepsilon (h^{-1}) +{\mathcal O}_\varepsilon(1)+{\mathcal O}(\frac{\varepsilon}{h}),
$$
where $\Gamma_0$ is a closed loop starting from $y_0$, and then returning to $y_0$ after having surrounded the anti-Stokes line $L_0^+$ between $y_+:=b_\infty^{-1/6}e^{i(2\arg\lambda -\pi)/6}$  and $y_0$ (this time, the anti-Stokes lines are those of  $i\sqrt{b_\infty}y^3 - \lambda_0$): see Figure 1. 

\begin{figure}[h]
\label{fig1}
\vspace*{0ex}
\centering
\includegraphics[width=0.6\textwidth,angle=0]{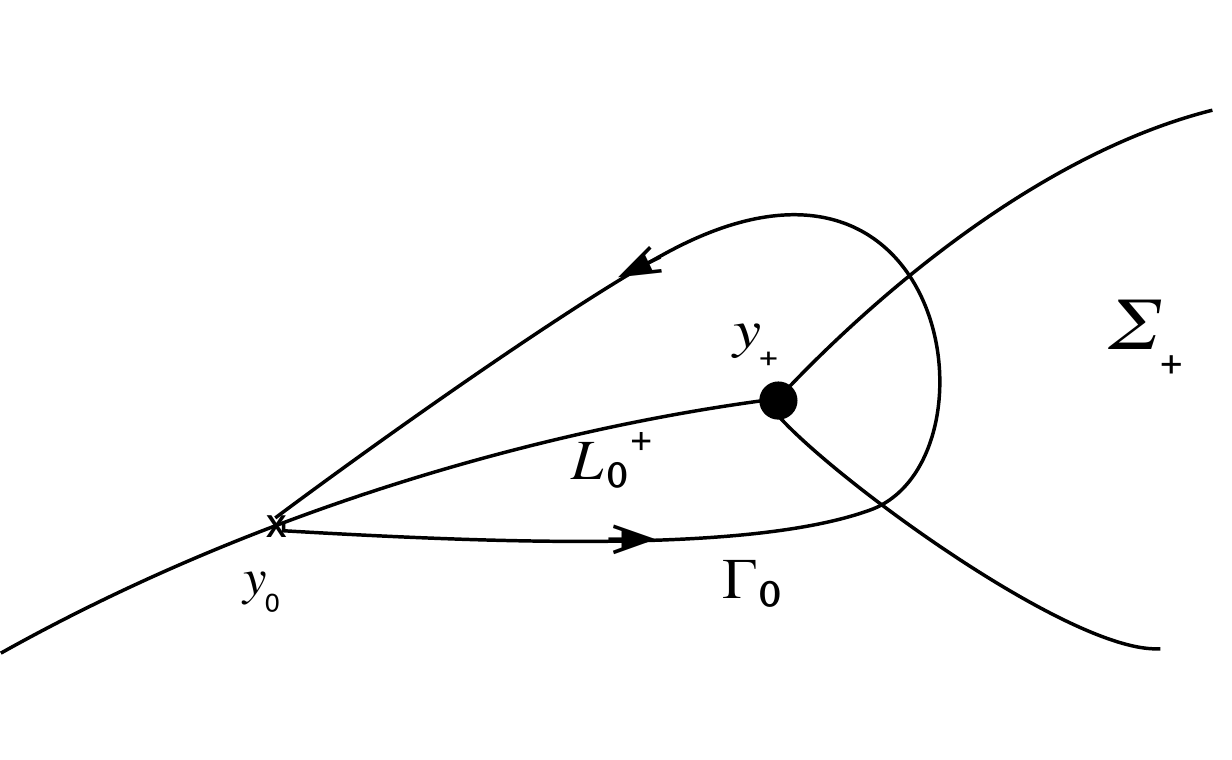}

\vspace*{0ex}
\caption{}
\end{figure}

Then, by deforming this contour up to $L_0^+$ (and because of the jump of the square root on it), a standard computation gives,
$$
\left| \frac1{2i\pi}\oint_{\Gamma_0}\sqrt{i\sqrt{b_\infty}y^3-\lambda_0}\, dy\right| =  \frac{2b_\infty^{\frac1{12}}}{\pi \sqrt 3}|y_+- y_0|^{3/2}+{\mathcal O}(|y_+- y_0|^{5/2})=: N_0.
$$
In particular, since $|y_+- y_0|\sim \delta$, we have  $N_0>0$  for  $\delta>0$ fixed sufficiently small. In this case,  we obtain,
\be
\label{asymptzero}
 N_D =\frac{N_0}{h} +o_\varepsilon (h^{-1}) +{\mathcal O}_\varepsilon(1)+{\mathcal O}(\frac{\varepsilon}{h}).
\ee
and thus, for fixed $\varepsilon>0$ small enough,  we see that $N_D$ tends to infinity as $h\to 0$ (that is, as $k\to \infty$). But this is in contradiction with  the fact that $N_D\leq N_\psi$  with $N_\psi$ finite and independent of $k$.
\end{proof}
\vskip 0.3cm
Since $H(\beta)$ forms an analytic family of type A with only finite multiplicity eigenvalues, and $E(\beta)$ remains bounded on $\Omega$, we have that, for any point $\beta_1$ of the boundary $\partial\Omega$, $E(\beta)$ can be extended (as an eigenvalue of $H(\beta)$) into a multiple-valued analytic  function on a neighborhood of $\beta_1$, with another possible branch point at $\beta_1$ only. \\

Now, if $\gamma\, :\, [0,1]\to \C$ is a simple continuous path, at a first stage we must consider the possibility of the existence of branch points (even if, afterwards, we will see that they do not exist). For that purpose, we define,

\begin{definition}\sl
For any $t_0\in (0,1)$, we call $\gamma${\bf -semi-neighborhood of} $\gamma(t_0)$ any open set $\omega\subset \C$ whose boundary $\partial\omega$ contains $\gamma(t_0)$ and coincides, near $\gamma(t_0)$, with  $\{\gamma(t); |t-t_0|<\varepsilon\} $ for some $\varepsilon >0$ (see Figure 2). 
\end{definition}

\begin{figure}[h]
\label{fig3}
\vspace*{0ex}
\centering
\includegraphics[width=0.6\textwidth,angle=0]{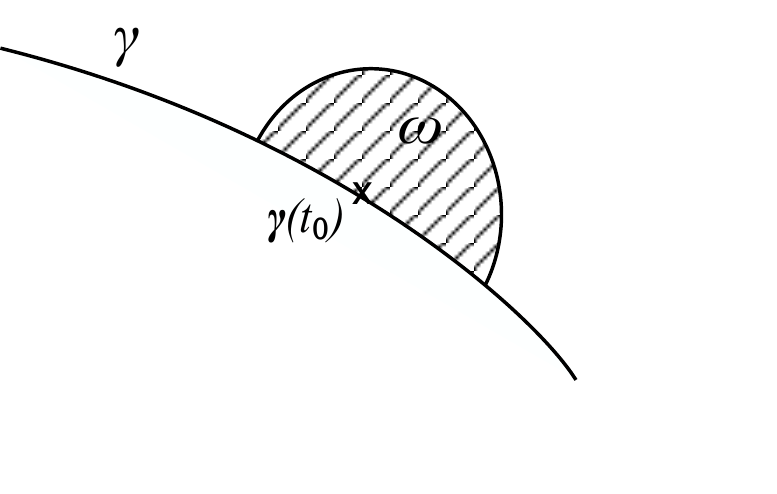}

\vspace*{0ex}
\caption{}
\end{figure}

Then, the previous discussion and a standard argument of connectedness show that we have,
\begin{proposition}\sl
\label{seminei}
Let $\gamma\, :\, [0,1]\to \C_c\backslash\{\beta_0\}$ be an arbitrary simple continuous path such that $\gamma(0)\in\Omega$. Then, 
for any determination $\tilde E(\beta)$ of $E(\beta)$, there exists a finite set $\{t_1,\dots,t_k\}\subset (0,1)$, such that $\tilde E(\beta)$ can be extended (as an eigenfunction of $H(\beta)$) into a holomorphic function on ${\mathcal W}:={\mathcal U}\cup \omega_1\cup\dots\cup\omega_k$, where ${\mathcal U}$ is a complex  neighborhood of $\gamma\backslash\{ \gamma(t_1),\dots,\gamma(t_k),\gamma(1)\}$, and $\omega_j$ ($j=1,\dots,k$) is a $\gamma$-semi-neighborhood of $\gamma(t_j)$. Moreover, the extension is continuous up to the boundary $\partial{\mathcal W}$ of ${\mathcal W}$.
\end{proposition}

In the next section, we relate such an extension with the $E_n(\beta)$'s of (\ref{calic}).\\

\begin{remark}\sl
Actually, the extension can be made in the larger open set obtained by replacing the  $\gamma$-semi-neighborhoods by ``cut-neighborhoods'' of the form $\tilde\omega_j\backslash L_j^+$, where $\tilde\omega_j$ is a neighborhood of $\gamma(t_j)$, and $L_j^+=\gamma(t_j) + z_j\R_+$ with $z_j\in \C\backslash\{0\}$ transverse to $\gamma$ at $\gamma (t_j)$ (in the case of a $C^1$-path). But the previous notion of semi-neighborhoods is enough for our purposes.
\end{remark}

\section{The spectrum for small parameter}
\label{secbetapeti}

Let $\gamma\, :\, [0,1)\to \C_c\backslash\{\beta_0\}$ be a simple continuous path, such that $\gamma(0)\in\Omega$, and $\{\gamma (t)\, ;\, 1-\mu \leq t <1\} = (0, \nu]$ for some $\mu,\nu >0$. \\

The previous proposition (applied to $\gamma\left\vert_{[0,1-\delta]}\right.$ for any $\delta >0$) shows the existence of a (possibly void) countable set $\{ t_j\,;\, 1\leq j\leq K\}\subset (0,1)$, with either $K<+\infty$, or $K=+\infty$ and $t_j\to 1$ as $j\to \infty$, such that any determination of $E(\beta)$ in $\Omega$ can be extended into a holomorphic function on ${\mathcal W}:={\mathcal U}\cup \omega_1\cup\dots\cup\omega_K$, where ${\mathcal U}$ is a complex neighborhood of $\gamma\backslash\{ t_j\, ;\, 1\leq j\leq K\}$, and $\omega_j$ ($j=1,\dots,K$) is a $\gamma$-semi-neighborhood of $\gamma(t_j)$.\\

Still denoting by $E(\beta)$ such an extension, and with $E_n(\beta)$ as in (\ref{calic}), we have,
\begin{proposition}\sl
\label{betapetit}
There exists an integer $n\geq 0$, such that, for all $\beta\in{\mathcal W}$ sufficiently small, one has,
$$
E(\beta)=E_n(\beta).
$$
\end{proposition}
\begin{remark} \sl In particular, $E(\beta)$ can be continued analytically in $\{ \beta\in\C_c\, ;\, |\beta| <\delta_n\}$, and thus, as a matter of fact,  the branch points cannot accumulate at $\beta=0$.
\end{remark}
\begin{proof} By construction,  ${\mathcal W}$ contains a sequence $(\beta_k)_{k\geq 0}\subset (0,\nu]$ that converges to 0.
We distinguish two cases.\\

{\bf Case 1}  : $|E(\beta_k)|$ does not tend to infinity as $k\to +\infty$.\\

In this case, by extracting a subsequence, we can assume that $E(\beta_k)$ admits a finite limit $E_0$ as $k\to \infty$. But since $H(\beta_k)$ tends to $H_0$ in the norm resolvent sense, we necessarily have $E_0\in \sigma (H_0)$, that is, there exists $n\geq 0$ such that $E_0=E_n(0)$. Using (\ref{calic})-(\ref{calicbis}), we deduce that $E(\beta_k) =E_n(\beta_k)$ for all $k$ sufficiently large. By the analyticity of $E(\beta)$ on ${\mathcal W}$, and the simplicity of $E_n(\beta)$ ($|\beta|<\delta_n$), we conclude that $E(\beta)=E_n(\beta)$ for all $\beta$ sufficiently close to some $\beta_k$ ($k>>1$), and thus (since $E_n(\beta)$ is analytic, too), for all $\beta\in {\mathcal W}$ sufficiently small.\\

{\bf Case 2} : $|E(\beta_k)|$ tends to infinity as  $k\to +\infty$.\\

We separate the elements of the sequence $B:=(\beta_k)_{k\geq 0}$ in two parts:
$$
{\mathcal P}_-:=\{ \beta \in B\, ;\, |E(\beta)|\leq \beta^{-1}\}\quad ;\quad {\mathcal P}_+:=\{ \beta \in B\, ;\, |E(\beta)|\geq \beta^{-1}\}.
$$
We will prove by contradiction that, actually, both of these sets are empty. \\

When $\beta\in {\mathcal P}_-$, we make the change of variable,
$$
x=\sqrt{|E(\beta)|}\, y,
$$
and we set,
\begin{eqnarray*}
&&\varphi (y) := |E(\beta)|^{\frac14}\psi_\beta (\sqrt{|E(\beta)|}\, y) \,\, ;\,\, h:= |E(\beta)|^{-1} (<<1)\\
&&  \alpha:=\sqrt{\beta |E(\beta)|}\in (0,1]\,\, ;\, \, \lambda :=\frac{E(\beta)}{|E(\beta)|}.
\end{eqnarray*}
Then, we have,
$$
-h^2\varphi''(y) + (y^2+i\alpha y^3 -\lambda)\varphi (y) =0,
$$
and we can use again complex WKB semiclassical methods. Since $\beta>0$, we have  $\Re H(\beta) =H_0 \geq 1$ and thus $\Re E(\beta) \geq 1$. Therefore, $|\arg \lambda |<\pi /2$ and, among the three turning points, two of them are located in $\{ \Im y \leq 1/\sqrt2 \}$, while the third one is in $\{ \Im y \geq 1\}$. Moreover, in this case Proposition \ref{zero1} gives an absence of zeros for $\varphi$ in the strip $\{ 0\leq\Im y\leq \frac3{2\sqrt{\beta|E(\beta)|}}\}$, and thus in $\{ 0\leq\Im y\leq 1\}$. Then, an argument similar to the one leading to (\ref{asymptzero}) shows that 
 $N_\psi$ tends to infinity as $\beta \to 0_+$, and this is in contradiction with Proposition \ref{nbrenode}.
\vskip 0.5cm

Now, when $\beta\in {\mathcal P}_+$, we make the change of variable,
$$
x= \frac{|E(\beta)|^{1/3}}{\beta^{1/6}}y,
$$
and we set,
\begin{eqnarray*}
&& \varphi (y) := \frac{|E(\beta)|^{1/6}}{\beta^{1/12}}\psi_\beta (\frac{|E(\beta)|^{1/3}}{\beta^{1/6}}\, y) \,\, ;\,\, h:= \frac{\beta^{1/6}}{|E(\beta)|^{5/6}} (<<1)\\
&& \alpha:= (\beta|E(\beta)|)^{-1/3}\in (0,1] \,\, ;\, \, \lambda :=\frac{E(\beta)}{|E(\beta)|}.
\end{eqnarray*}
In this case, $\varphi$ is solution of,
$$
-h^2\varphi''(y) + (\alpha y^2+iy^3 - \lambda)\varphi (y) =0.
$$
Again, complex semiclassical WKB asymptotics can be used, and we obtain  that $N_\psi$ should tend to infinity as $\beta\to 0_+$, which is in contradiction with Proposition \ref{nbrenode}. (Let us observe that, by an obvious change of variable, in both cases the turning points can be deduced from those of $y^2+iy^3-\lambda'$, with $\Re\lambda' >0$.)
\end{proof}

\section{Holomorphic extension}

In view of Proposition \ref{betapetit},  the first part of Theorem \ref{mainth} will follow from the next result.\\
\begin{proposition}\sl
\label{holomext}
For any $n\geq 0$, the function,
$$
\C_c\cap\{ |\beta|  <\delta_n\}\ni \beta \mapsto E_n(\beta)
$$
can be extended  to a holomorphic function on the whole cut plane $\C_c$, in such a way that, for all $\beta\in\C_c$ the value at $\beta$ of the extension (still denoted by $E_n(\beta)$) is a simple eigenvalue of $H(\beta)$.
\end{proposition}
\begin{proof} It is enough to prove that $E_n(\beta)$ can be continued, as a simple eigenvalue of $H(\beta)$, along any simple continuous path $\gamma\, :\, [0,1]\to \C_c$ such that $|\gamma (0)|<\delta_n$.\\

By Proposition \ref{seminei}, we already know that $E_n(\beta)$ can be extended, as an eigenvalue of $H(\beta)$, into a continuous function on $\gamma$ (and, in fact, holomorphic at least on `one side' of $\gamma$).
We set,
$$
t_0:= \sup\{ t\in  [0,1]\, ;\, E_n(\beta) \mbox{ is simple at } \gamma (s) \mbox{ for all } s\in[0,t]\}.
$$
By contradiction, let us assume $t_0<1$. \\

Since $E_n(\beta)$ is continuous along $\gamma$, this implies that $E_n(\beta_0)$ is a multiple eigenvalue of $H(\beta_0)$, where we have set $\beta_0:=\gamma (t_0)$. Hence, because of the analyticity of the family $H(\beta)$, in a neighborhood of $\beta_0$ the spectrum of $H(\beta)$ is given either by at least two branches of a multi-valued analytic function around $\beta_0$, or by at least two different analytic functions (see \cite{K}, Chapter VII).\\

In particular, for  $\beta$ in a sufficiently small  neighborhood $\Omega_1$ of $\{\gamma (t)\, ;\, 0<t_0-t_1<<1\}$, there exists an eigenvalue $E(\beta)$, depending analytically on $\beta$, such that 
 $E(\beta)\not= E_n(\beta)$ for all $\beta\in\Omega_1$, and $E(\beta)\to E_n(\beta_0)$ as $\beta\to\beta_0$ (recall that, by definition, $E_n(\gamma (t))$ is simple for $t<t_0$).\\
 
On the other hand, by Sections \ref{prelim} and \ref{secbetapeti}, we know that $E(\beta)$ can be extended, as an eigenvalue of $H(\beta)$, into a holomorphic function on the neighborhood of some simple continuous path connecting $\Omega_1$ with a neighborhood of 0,
and that there exists some $m\geq 0$ such that $E(\beta) =E_m(\beta)$ for $|\beta|$ sufficiently small. \\

Then, by construction, we necessarily have $m\not= n$, otherwise,  by the simplicity (and thus, the analyticity) of $E_n(\beta)$ on $\{ \gamma (t)\, ;\, 0\leq t<t_0\}$, we would have $E(\beta)=E_n(\beta)$ on this set.\\

Now, we denote by $\psi_{n,\beta}$ and $\psi_{m,\beta}$ two normalized eigenfunctions associated with $E_n(\beta)$ and $E(\beta)=E_m(\beta)$, respectively, and depending continuously on $\beta$ (see Proposition \ref{contfonctpropre}). By Proposition \ref{nbrenode}, we know that their respective number of nodes, say $N_n$ and $N_m$, do not depend on $\beta$. Moreover, since $E(\beta_0) =E_n(\beta_0)$ and Ker$(H(\beta_0) - E_n(\beta_0))$ is one dimensional, we must have $\psi_{n,\beta_0}$ and $\psi_{m,\beta_0}$ co-linear, and thus, necessarily, $N_n=N_m$. Therefore, in order to obtain a contradiction, it is enough to prove,
\begin{lemma}\sl
\label{nnodes}
For any $n\geq 0$, the eigenfunction associated to $E_n(\beta)$ ($\beta\in\C_c$, $|\beta|<\delta_n$) admits exactly $n$ nodes.
\end{lemma}
\begin{proof} Denote by $\varphi_n$ the $(n+1)$-th normalized eigenfunction of the harmonic oscillator $H_0$. It is well known that $\varphi_n$ is holomorphic on $\C$, and admits exactly $n$ zeros that are all located in the real interval $I_n:=[-\sqrt{E_n(0)}, \sqrt{E_n(0)}]$. Moreover, since $H(\beta)$ tends to $H_0$ in the norm resolvent sense as $\beta\to 0$, and $E_n(0)$ is simple, the conveniently normalized eigenfunction $\psi_{n,\beta}$ associated to $E_n(\beta)$ ($|\beta|<\delta_n$) converge in norm towards $\varphi_n$ as $\beta\to 0$. By the ellipticity of $H_0$ and $H(\beta)$ this actually implies that $\psi_{n,\beta}-\varphi_n$ tends to 0 in any Sobolev norm, and thus locally uniformly (together with its derivative). As a consequence, if $C$ is an arbitrarily large positive constant,  for all $|\beta|$ small enough one has $\psi_{n,\beta}(z)\not=0$ if $C^{-1}\leq {\rm dist}\, (z, I_n) \leq C$, and the number $N$ of zeros of $\psi_{n,\beta}$ in $\{ {\rm dist}\, (z, I_n) \leq C^{-1}\}$ is given by,
\begin{eqnarray*}
&&N=\frac1{2i\pi}\oint_{\{ {\rm dist}\, (z, I_n) = C^{-1}\}} \frac{\psi_{n,\beta}'(x)}{\psi_{n\beta}(x)}dx\\
&& \hskip 0.5cm \longrightarrow  \frac1{2i\pi}\oint_{\{ {\rm dist}\, (z, I_n) = C^{-1}\}} \frac{\varphi_n'(x)}{\varphi_n(x)}dx =n\quad (\beta\to 0).
\end{eqnarray*}
Therefore, $N=n$ and, by Proposition \ref{nbrenode}, these zeros are exactly the nodes of $\psi_{n,\beta}$.
\end{proof}
It follows from Lemma \ref{nnodes} and the previous discussion that $R_\varepsilon=+\infty$. Since $\varepsilon >0$ is arbitrary, we conclude that $E_n(\beta)$ remains simple for all $\beta\in\C_c$, and Proposition \ref{holomext} is proved.
\end{proof}

Putting together Proposition \ref{betapetit} and Proposition \ref{holomext}, we obtain the fact that, for any $\beta\in\C_c$, the spectrum of $H(\beta)$ is simple, and consists exactly of the holomorphic extensions of the perturbative eigenvalues $E_n(\beta)$'s defined for $|\beta|$ small. In the rest of the paper, we will prove the last part of Theorem \ref{mainth}, that is, the Pad\'e summability of the $E_n(\beta)$'s. In order to do so, we first need to know the asymptotic behaviours of $E_n(\beta)$ as $\arg \beta \to (\pm \pi)_{\mp}$ and as $|\beta|\to \infty$.

\section{Beheviour on the cut}\label{seccut}

From now on, we fix $n\geq 0$, and we consider the simple eigenvalue $E_n(\beta)$ defined for all $\beta\in\C_c$. 
We also observe that, since $H(\beta )$ admits the symmetry,
$$
{\mathcal P}{\mathcal T} H(\beta) =H(\overline{\beta}){\mathcal P}{\mathcal T},
$$
(where we have denoted by ${\mathcal P}\psi(x):=\psi(-x)$ the parity operator, and by ${\mathcal T}\psi =\overline{\psi}$ the
time reversal), the simplicity of $E_n(\beta)$ and the characterization by the number of nodes immediately give,
\be
\label{ptsym}
E_n(\overline{\beta})= \overline{E_n(\beta)}
\ee
for all $\beta\in \C_c$. In particular, for $\beta>0$ we recover the fact that the spectrum of $H(\beta)$ is real (and thus included in $[1, +\infty)$, since $\Re H(\beta) \geq 1$).
\vskip 0.3cm

By the proof of Proposition \ref{bounded} and (\ref{calicbis}), we also know that, for any $C>0$, $E_n(\beta)$ remains uniformly bounded on $\{ \beta\in\C_c\, ;\, |\beta|\leq C\}$. In this section, we study more precisely the behavior of $E_n(\beta)$ as $\beta$ approaches the cut $(-\infty, 0]$ of $\C_c$.

\begin{proposition}\sl
\label{behavcut}
For any $b>0$, $E_n(\beta)$ admits a finite limit $E_n^{\pm}(-b)$ as $\beta\to -b$, $\pm\Im \beta>0$, and
the  $E_n^{\pm}(-b)$'s form the set of  eigenvalues (all simple) of the operator,
$$
\tilde H_{\pm \alpha}(-b):= -e^{\pm 2i\alpha}\frac{d^2}{dx^2}+e^{\mp 2i\alpha}x^2 \mp \sqrt{b}\,e^{\mp 3i\alpha}x^3
$$
where $\alpha$ is any sufficiently small positive number, and the domain ${\mathcal D}$ of $\tilde H_{\pm \alpha}(-b)$ is the same as the one of $H(\beta)$. Moreover, one has,
$$
\pm \Im E_n^{\pm}(-b) >0.
$$
 \end{proposition}
 \begin{proof}
We prove it for $\Im \beta >0$ (the case $\Im\beta <0$ is completely analogous, and also results from (\ref{ptsym})). In that case, for any fixed $\alpha>0$ small enough ($\alpha<\pi/5$ is enough), the complex change of variable $x=e^{-i\alpha}y$ transforms $H(\beta)$ into,
$$
\tilde H_\alpha (\beta):= -e^{2i\alpha}\frac{d^2}{dy^2} + e^{-2i\alpha}y^2 +i\sqrt{\beta}e^{-3i\alpha}y^3.
$$
and the  eigenfunction $\psi_{n,\beta}(x)$, associated with $E_n(\beta)$, becomes,
$$
\varphi_{n,\beta}(y):= \psi_{n,\beta}(e^{-i\alpha}y).
$$
(Observe that, with  $S_{\pm 1}(\beta)$ defined in (\ref{secteurs}),  when $y$ is real we have $e^{-i\alpha}y\in S_{\pm 1}(\beta)$ for all $\beta$ such that $|\pi -\arg\beta|$ is small enough, without restriction on the sign of $(\pi -\arg\beta)$.)\\

Writing $\beta =b'e^{i\theta}$ with $|\pi -\theta |<< 1$, we have,
\be
\label{Htildebis}
\tilde H_\alpha (\beta):= e^{2i\alpha}(-\frac{d^2}{dy^2} + e^{-4i\alpha}y^2 -\sqrt{b'}e^{-i(\frac{\pi-\theta}2+5\alpha)}y^3),
\ee
and thus, since $\frac{\pi-\theta}2+5\alpha\geq 5\alpha >0$, it is not difficult to see (e.g., as in \cite{CGM}) that $\tilde H_\alpha (\beta)$ has a compact resolvent (even for $\arg\beta =\pi$), that it forms an analytic family of type A around $\beta=-b$, and that $\tilde H_\alpha (\beta)$ tends to $\tilde H_\alpha (-b)$ in the norm resolvent sense, as $\beta\to -b$, $\Im\beta >0$. As a consequence, any cluster point of $E_n(\beta)$ is necessarily in the spectrum of $\tilde H_\alpha (-b)$. Since in addition $E_n(\beta)$ is simple when $\Im\beta >0$ and remains uniformly bounded as $\beta\to-b$, $\Im \beta >0$, we conclude that it can admit only one cluster point, that is, a limit $E_n^+(-b) \in \sigma (\tilde H_\alpha (-b))$, as $\beta\to-b$, $\Im \beta >0$.\\

Moreover, by considering the number of nodes of $\psi_{n,\beta}$, we see as before that all the eigenvalues of $\tilde H_{\pm\alpha}(\beta)$ are simple and depend analytically on $\beta$ near $-b$.\\

Concerning the sign of $\Im E_n^+(-b)$, we see on the expression (\ref{Htildebis}) that the  eigenfunction $\varphi_n^+$ of $\tilde H_\alpha (-b)$ associated with $E_n^+(\alpha)$ satisfies to the asymptotics (see \cite{Si}, Chapter 2),
$$
\frac{(\varphi_n^+)'(y)}{\varphi_n^+(y)} =  -b^{1/4}e^{i\frac{\pi - 5\alpha}2}y^{3/2}(1+o(1)),
$$
as $\Re y\to +\infty$ with $|\arg y|< \pi/5$. In particular, this behavior is valid for $y=e^{i\alpha}x$ with $x$ real, $x\to +\infty$, and in this case, setting $\psi_n^+(x):= \varphi_n^+(e^{i\alpha}x)$, we obtain that $\psi_n^+$ is a solution of the equation,
\be
\label{eqlim}
-(\psi_n^+)'' + (x^2 - \sqrt{b} \,x^3)\psi_n^+ =E_n^+(-b))\psi_n^+,
\ee
with the asymptotic behavior,
\be
\label{asymppsi+}
\frac{(\psi_n^+)'(x)}{\psi_n^+(x)} =  -ib^{1/4}x^{3/2}(1+o(1)),\quad (x\to +\infty).
\ee
Moreover, for the same reasons (and since the potential $x^2 - \sqrt{b} \,x^3$ tends to $+\infty$ as $x\to -\infty$), we see that $\psi_n^+(x)$ is exponentially small (together with all its derivatives) as $x\to -\infty$.
Then, using the equation (\ref{eqlim}), we have,
$$
\Im E_n^+(-b) \int_{-\infty}^x|\psi_n^+(t)|^2 dt= - \Im (\psi_n^+)'(x)\overline{\psi_n^+(x)}
$$
where $x>0$ is arbitrary, and thus, by (\ref{asymppsi+}),
$$
\Im E_n^+(-b) \int_{-\infty}^x|\psi_n^+(t)|^2 dt= b^{1/4}x^{3/2}|\psi_n^+(x)|^2(1+o(1)) \quad (x\to +\infty).
$$
In particular, taking $x>0$ sufficiently large, we obtain $\Im E_n^+(-b)  >0$.
\end{proof}

Let us also observe that, still by \cite{Si}, Chapter 2, one also has  $|\psi_n^+(x)|^2 \sim x^{-3/2}$ as $x\to +\infty$, so that, actually, $\psi_n^+ \in L^2(\R)$ (but $\notin {\mathcal D}$).

\section{Behaviour at infinity}

Thanks to the previous section (plus the results of \cite{CGM} for $\beta$ small), we can extend $E_n(\beta)$ in a continuous way up to $\arg\beta =\pm\pi$, by setting,
$$
E_n(be^{\pm i\pi}) := E_n^{\pm} (-b).
$$
(Of course, in this notation, the quantity $e^{\pm i\pi}$ cannot be replaced by $-1$.)

\vskip 0.3cm

In this section, we study the behavior of $E_n(\beta)$ as $|\beta| \to +\infty$, $-\pi \leq \arg \beta \leq \pi$.
We have,
\begin{proposition}\sl
\label{behavinfty}
For any $n\geq 0$, the quantity $\beta^{-1/5}E_n(\beta)$ admits a finite limit $L_n\geq 0$ as $|\beta| \to +\infty$, $-\pi \leq \arg \beta \leq \pi$. Moreover, if one sets,
$$
H_\infty := -\frac{d^2}{dx^2} + i x^3
$$
with domain ${\mathcal D}$, then, the spectrum of $H_\infty$ consists exactly of the set $\{L_n\,;\, n\geq 0\}$, and $L_n$ is the $(n+1)$-th eigenvalue of $H_\infty$.
\end{proposition}
\begin{proof}
The complex change of variable $x=\beta^{-\frac1{10}}y$ transforms $H(\beta)$ into,
$$
\hat H (\beta):= \beta^{1/5}(-\frac{d^2}{dy^2} + \beta^{-2/5}y^2 + iy^3).
$$
Moreover, as in Section \ref{seccut}, we see that the eigenfunction $\psi_{n,\beta}$ associated with $E_n(\beta)$ is transformed into a function in the domain of the operator. As a consequence, setting,
$$
\check H(\alpha):=-\frac{d^2}{dy^2} + \alpha y^2 + iy^3,
$$
we see that, for $|\arg\alpha|\leq 2\pi/5$, the spectrum of $\check H(\alpha)$ consists of the eigenvalues $\alpha^{1/2}E_n(\alpha^{-5/2})$, $n=0,1,\dots$, with eigenfunctions $\varphi_{n,\alpha}(y):=\psi_{n,\alpha^{-5/2}}(\alpha^{1/4}y)$ admitting exactly $n$ zeros in $\{ -\pi\leq \arg y+ \frac14\arg\alpha\leq 0\}$ (nodes). 
In addition, for $\alpha\in\C$ small enough, $\check H(\alpha)$ forms an analytic family of type A with compact resolvent (see, e.g., \cite{CGM}), and is a small perturbation of $\check H(0) = -\frac{d^2}{dy^2} + iy^3$ (in the sense that it tends to $\check H(0)$ in the norm resolvent sense as $\alpha\to 0$).
\vskip 0.3cm
Then, by arguments similar to (but, somehow, simpler than) those used in the proof of Proposition \ref{bounded}, we see that $\alpha^{1/2}E_n(\alpha^{-5/2})$ remains bounded as $\alpha\to 0$, $|\arg\alpha|\leq 2\pi/5$. Moreover, any limit value is an eigenvalue of $\check H(0)$, and since its spectrum is discrete, we conclude that $\alpha^{1/2}E_n(\alpha^{-5/2})$ necessarily admits a limit $L_n\in\sigma (\check H(0))$ as $\alpha\to 0$, $|\arg\alpha|\leq 2\pi/5$. Moreover, any eigenvalue of $\check H(0)$ gives rise, for $\alpha$ small, to some eigenvalue of $\check H(\alpha)$.\\

Then, specifying the way in which $\alpha$ tends to zero by taking $\alpha >0$, and using the fact that, in this case, $\alpha^{1/2}E_n(\alpha^{-5/2})$ is a positive number and its associated eigenfunction admits exactly $n$ nodes, we first obtain that $L_n\geq0$, then that it is simple, and finally that it is the $(n+1)$-th eigenvalue of $\check H(0)$.
\end{proof}
\section{Pad\'e Summability}

Now, we come to the last part of the proof of Theorems \ref{mainth} and \ref{mainth2}. We first have,
\begin{proposition}\sl
\label{stieltfunc}
For $n\geq 0$ fixed and $\beta\in\C_c$, set,
$$
F_n(\beta):= \frac{E_n(\beta) -E_n(0)}{\beta}.
$$
Then, $F_n$ is a Stieltjes function, that is, more precisely, it can be written on the form,
$$
F_n(\beta) =\int_0^{+\infty}\frac{\rho_n(t)}{1+\beta t} dt,
$$  
where $\rho_n$ is a  real-analytic positive function such that, for all $N\geq 0$, one has  $t^N\rho_n(t) \in L^1(\R_+)$.
\end{proposition}
\begin{proof} Since $F_n$ is holomorphic on $\C_c$, for any $\beta\in\C_c$ we have,
$$
F_n(\beta ) = \frac1{2i\pi}\oint_\gamma \frac{F_n(z)}{z-\beta} dz,
$$
where $\gamma$ is an arbitrary simple loop in $\C_c$ surrounding $\beta$ (and positively oriented). Moreover, if $|z|=R>>|\beta|$ and $|\arg z|< \pi$, by Proposition \ref{behavinfty}, we have,
$$
\left| \frac{F_n(z)}{z-\beta}\right| \leq C_n \frac{R^{-4/5}}{R-|\beta|}
$$
with $C_n >0$ constant, and thus,
$$
\int_{|z|=R, |\arg z|<\pi}\left|\frac{F_n(z)}{z-\beta}\right| \,| dz| \longrightarrow 0 \quad \mbox{as }\,\, R\to +\infty.
$$
Therefore, we can first deform $\gamma$ up to a contour $\tilde \gamma = \gamma_+ \cup \gamma_-$, where $\gamma_+$ follows $(-\infty, 0]$ on the side $\{\Im z > 0\}$ (and is oriented from $-\infty$ to $0$), while $\gamma_-$ follows $(-\infty, 0]$ on the side $\{\Im z < 0\}$ (and is oriented from $0$ to $-\infty$): see Figure 2. 
\begin{figure}[h]
\label{fig2}
\vspace*{0ex}
\centering
\includegraphics[width=0.6\textwidth,angle=0]{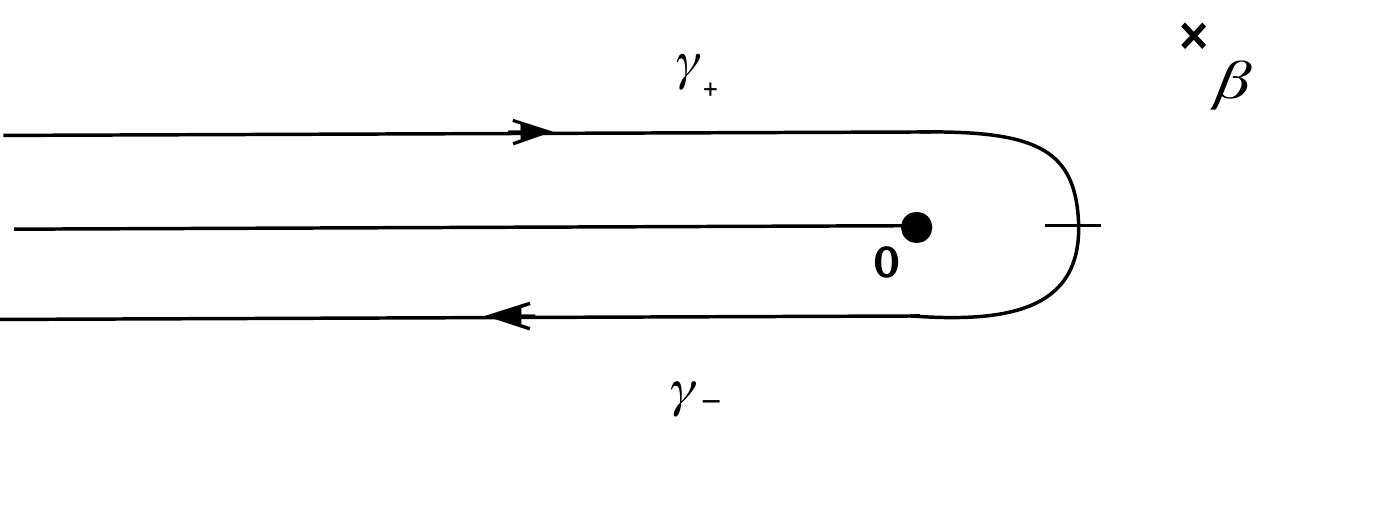}

\vspace*{0ex}
\caption{}
\end{figure}

Then, using Proposition \ref{behavcut}, we can take the limit where both $\gamma_+$ and $\gamma_-$ become the interval $(-\infty,0]$, and we obtain,
$$
F_n(\beta )  = \frac1{2i\pi}\int_{-\infty}^0\frac{F^+_n(z)-F^-_n(z)}{z-\beta} dz,
$$
where we have used the notation,
$$
F_n^{\pm}(z):= \frac{E_n^{\pm}(z) -E_n(0)}{z}\quad (z\in (-\infty,0]).
$$
Therefore, setting $t:= -z^{-1}$, this gives us,
$$
F_n(\beta) =\int_0^{+\infty}\frac{\rho_n(t)}{1+\beta t} dt
$$
with,
$$
\rho_n(t):=\frac1{2i\pi t}(F_n^-(-\frac1t)-F_n^+(-\frac1t))=\frac1{2i\pi}(E_n^+(-\frac1t)-E_n^-(-\frac1t)).
$$
By  (\ref{ptsym}), we also have $E_n^-(-\frac1t) =\overline{E_n^+(-\frac1t)}$, and thus, using Proposition \ref{behavcut} again, we obtain,
\be
\label{rhon}
\rho_n(t) =\frac1{\pi}\Im E_n^+(-\frac1t) >0.
\ee
In particular, the real-analyticity of $\rho_n$ is a direct consequence of Proposition \ref{behavcut} and its proof. \\

Finally, the fact that $t^N\rho_n(t)\in L^1(\R_+)$ comes from Proposition \ref{behavinfty} and (\ref{calicbis}). Indeed, they tell us that $\rho_n(t)\sim t^{-1/5}$ as $t\to 0_+$, and $\rho_n(t)={\mathcal O}(t^{-\infty })$ as $t\to +\infty$.
\end{proof}

Now, we are able to complete the proof of Theorem \ref{mainth}. By Proposition \ref{stieltfunc}, as $\beta\to 0_+$,  $F_n(\beta)$ admits the asymptotic expansion,
$$
F_n(\beta)\sim \sum_{j\geq 0}a_{n,j} (-\beta)^j,
$$
with,
\be
\label{rhosol}
a_{n,j} := \int_0^{+\infty} t^j \rho_n(t) dt\,\,  (>0).
\ee

On the other hand, by (\ref{calicbis}), we also have $a_{n,j} = |e_{n,j+1}|$, and thus, by (\ref{anasymb}),
$$
a_{n,j}\leq D_n C_n^{j+1}(j+1)!.
$$
In particular,
$$
\sum_{j=1}^{+\infty} (1/a_{n,j})^{1/2j} = +\infty.
$$
This means that the criterion of Carleman Theorem (see \cite{Wa}, Theorem 88.1, page 330) is satisfied, and thus the (possible) solution $d\mu$ ($=$ non negative measure on $\R_+$) of the moment problem,
$$
 \int_0^{+\infty} t^j d\mu (t) =a_{n,j}
$$
is unique. But, by (\ref{rhosol}), we already know that such a measure exists, and is given by $d\mu (t) =\rho_n(t)\, dt$. At this point, we can apply the results of \cite{Wa}, Chapter XIX (in particular, Theorem 97.1), together with \cite{St1}, Chapters VII and VIII (in particular \S 47,48, 51 and 54). They tell us that the series $\sum_{j\geq 0}a_{n,j} z^j$ is Stieltjes (see \cite{Wa}, \S 97), that the diagonal Pad\'e approximants $(P_{n,j}(\beta), Q_{n,j}(\beta))_{j\geq 0}$ of $\sum_{j\geq 0}a_{n,j} (-\beta)^j$, with $Q_{n,j}(0)=1$, exist and are unique, that they have no zeros in $\C_c$, and that, for all $\beta\in\C_c$, the quotient $P_{n,j}(\beta)/Q_{n,j}(\beta)$ tends to $\int_0^{+\infty}\frac{\rho_n(t)}{1+\beta t} dt= F_n(\beta )$, as $j\to +\infty$.
\vskip 0.3cm 
Hence, Theorem \ref{mainth} is proved, together with the first part of Theorem \ref{mainth2}. It remains to estimate $\ln \rho_n(t)$ as $t\to +\infty$. We observe that, by (\ref{rhon}), this is equivalent to find an estimate on $\Im E_n^+(-b)$ as $b\to 0_+$. But then, the general theory of \cite{HS} tells us that,
$$
\Im E_n^+ (-b) = p_nb^{-q_n}e^{-A/b}(1+o(1)) \quad(b\to 0_+)
$$
where $p_n, q_n$ are positive numbers independent of $b$, and $A$ corresponds to the tunneling through the barrier, and is given by,
$$
A=2\int_0^{1}\sqrt{y^2-y^3}\, dy= \frac{8}{15}
$$
(observe that the change of variable $x=b^{-1/2}y$ transforms the operator $-d_x^2+x^2-\sqrt{b}x^3$ into $\frac1{b}[-b^2d_y^2+y^2-y^3]$, leading to a semiclassical problem as $b\to 0_+$). Then, the asymptotics  of $\rho_n(t)$ at $+\infty$  immediately follows by taking the logarithm. \\

Concerning the value of $C_n$ in (\ref{anasymb}), it also follows by writing,
\begin{eqnarray*}
|e_{n,k+1}| &=& \int_0^{+\infty} t^k\rho_n(t)dt = \int_0^{+\infty}p_nt^{k+q_n}e^{-At}dt +R_n(k) \\
&=& \frac{p_n}{A^{k+q_n}}\Gamma (k+q_n+1)+R_n(k) ,
\end{eqnarray*}
where, for any $C>0$,  $R_n(k)$ can be written as,
$$
R_n(k) = \int_0^C{\mathcal O}(t^k)dt + \int_C^{+\infty} o(t^{k+q_n})e^{-At}dt,
$$
where the estimates are uniform with respect to $k$. Hence, we have $R_n(k) =o(A^{-k}\Gamma (k+q_n+1))$ as $k\to +\infty$, and the result follows.
$\square$

\section{Appendix}

Here, we give a proof of Lemma \ref{estphi'surphi}.
\vskip 0.3cm
For any $C>0$ large enough, we set,
$$
\Omega_C(h):= \{ x\in\C\,; |\varphi' (x) | > \frac{C}{h} |\varphi (x) |\}\cap \Omega_0,
$$
where $\Omega_0$ is some fixed arbitrarily large bounded open subset of $\C$.
Then, $\Omega_C(h)$ is open and $\varphi'(x)\not = 0$ on $\Omega_C(h)$. \\

Our purpose is to prove that, for any $\delta >0$ fixed, if $C$ is sufficiently large, then $\Omega_C(h)\subset \{\dist (x,\varphi^{-1}(0))<\delta h\}$.\\

For $x\in \Omega_0$ such that $\varphi'(x)\not=0$, we define,
$$
u(x):=  \frac{\varphi(x)}{\varphi'(x)}.
$$
We have,
$$
u' = 1-\frac{\varphi'' \varphi}{\varphi^2} = 1-\frac{W}{h^2}u^2,
$$
where $W=W_h(x)$ is bounded together with all its derivatives on $\Omega_0$, uniformly with respect to $h$. In particular, since $|u|\leq h/C$ on $\Omega_C(h)$, on this set we obtain,
$$
|u'-1|\leq \frac{C_0}{C^2},
$$
with $C_0:=\sup_{\Omega_0}|W|$. 

Now, let $x_0=x_0(h)\in \Omega_C(h)$  arbitrary, and set,
\begin{eqnarray*}
&&v(t):= \frac1{h}u(x_0+th)\quad ;\quad f(t):= t-v(t);\\
&&\tilde\Omega_C:= \{ t\in \C\, ;\, x_0+th\in\Omega_C\}=\{t\, ;\, |v(t)|<\frac1{C}\}.
\end{eqnarray*}
The previous estimates give,
\be
\label{v'v''}
|v'(t)-1|=|f'(t)|={\mathcal O}(\frac1{C^2}) \,\mbox{ on } \tilde\Omega_C,
\ee
uniformly with respect to $C$ and $h$.\\

In order to prove that, for $C$ large enough, $x_0$ is distant less than $\delta h$ from $\varphi^{-1}(0)$, we plan to apply the fixed-point theorem to $f(t)$ on the open set ${\mathcal V}_C:=\{ |t|\leq 2/C\}$. Thus, we first have to prove that $f$ sends this set into itself. For  $\mu\in[0,2]$, let us set,
$$
S_\mu := \sup_{|t|\leq \mu/C}|v(t)|.
$$
Since $v' = 1-Wv^2$, for all $t$ such that $|t|\leq \mu/C$, we have,
$$
|v(t)|\leq |v(0)| + \frac{\mu}{C}(1+C_0S_\mu^2)
$$
and thus,
$$
S_\mu \leq \frac3{C}+\frac{2C_0}{C}S_\mu^2.
$$
Moreover, $S_\mu$ depends continuously on $\mu$, and $S_{\mu=0}=|v(0)|\leq 1/C$. Hence, for $C$ is sufficiently large, we necessarily have,
$$
S_\mu \leq \frac{C}{4C_0}(1-\sqrt{1-24(C_0/C^2)})< \frac{12}{C}.
$$
In particular, for $\mu=2$, this means that ${\mathcal V}_{C} \subset \tilde\Omega_{C/12}$, and  thus, by (\ref{v'v''}), we have $|f'|={\mathcal O}(1/C^2)$ on this set. As a consequence, if $t\in{\mathcal V}_C$, one has,
$$
|f(t)| =|f(0)|+{\mathcal O}(\frac{|t|}{C^2})\leq \frac1{C}+{\mathcal O}(\frac{1}{C^3}),
$$
and thus, for $C$ large enough,
$$
|f(t)| <\frac2{C}.
$$
This proves that $f$ sends ${\mathcal V}_{C}$ in itself. In addition, by (\ref{v'v''}) (and the fact that ${\mathcal V}_{C} \subset \tilde\Omega_{C/12}$), $f$ is also a contraction on ${\mathcal V}_{C}$. As a consequence, it admits a  fixed point in this set, and this means that there exists a zero of $\varphi$ distant from $x_0$ less that $2h/C$.
$\square$

\end {document}